\documentclass[10pt, a4paper]{article}

\usepackage[dvips]{graphicx}
\usepackage[top=2.4cm, bottom=2.4cm, left=2cm, right=2cm]{geometry}
\usepackage{multirow}
\usepackage{todonotes}
\usepackage{amssymb,amsmath,amsthm}
\usepackage{stmaryrd}
\usepackage[ruled, linesnumbered]{algorithm2e}
\SetKw{Break}{break}
\usepackage{tikz}
\usepackage{graphicx}
\usepackage[caption=false]{subfig}
\usepackage[space]{grffile}

\usepackage{url}
\usepackage{xspace}

\usepackage{amsmath,amsfonts,amsthm,paralist}
\newtheorem{theorem}{Theorem}
\newtheorem{lemma}{Lemma}

\newtheorem{definition}{Definition}
\newtheorem{assumption}{Assumption}

\newcommand\resource[2]{\ensuremath{p_{#1}^{(#2)}}\xspace}
\newcommand\resources[1]{\ensuremath{p_{#1}}\xspace}
\newcommand\allresources{\ensuremath{\mathbf{p}}\xspace}
\newcommand\allstarttimes{\ensuremath{\mathbf{s}}\xspace}
\newcommand\maxresource[1]{\ensuremath{P^{(#1)}}\xspace}
\newcommand\exectime[1]{\ensuremath{t_{#1}}\xspace}

\newcommand\opt{\ensuremath\textsc{opt}\xspace}

\DeclareMathOperator*{\argmin}{arg\,min}

\newenvironment{proofsketch}{%
  \proof}{\endproof}

\SetKwRepeat{Do}{do}{while}
\SetKw{Break}{break}
\SetKwFor{When}{when}{do}{end}

\newcommand{\tabincell}[2]{\begin{tabular}{@{}#1@{}}#2\end{tabular}}

\def\abstract{\vspace{.5em}
{\textit{\bf Abstract}:\,\relax%
}}

\def\keywords{\vspace{.5em}
{\textit{\bf Keywords}:\,\relax%
}}
\def\endkeywords{\par}

\begin{document}

\title{Multi-Resource List Scheduling of Moldable Parallel Jobs under Precedence Constraints}
\author{Lucas Perotin$^\S$, Hongyang Sun$^\dag$, Padma Raghavan$^\dag$ \\
\\
$^\S$Laboratoire LIP, ENS Lyon, Lyon, France\\
$^\dag$Vanderbilt University, Nashville, TN, USA\\
lucas.perotin@ens-lyon.fr; \{hongyang.sun, padma.raghavan\}@vanderbilt.edu
}
\date{}
\maketitle

\begin{abstract}
The scheduling literature has traditionally focused on a single type of resource (e.g., computing nodes). However, scientific applications in modern High-Performance Computing (HPC) systems process large amounts of data, hence have diverse requirements on different types of resources (e.g., cores, cache, memory, I/O). All of these resources could potentially be exploited by the runtime scheduler to improve the application performance. In this paper, we study multi-resource scheduling to minimize the makespan of computational workflows comprised of parallel jobs subject to precedence constraints. The jobs are assumed to be moldable, allowing the scheduler to flexibly select a variable set of resources before execution. We propose a multi-resource, list-based scheduling algorithm, and prove that, on a system with $d$ types of schedulable resources, our algorithm achieves an approximation ratio of $1.619d+2.545\sqrt{d}+1$ for any $d$, and a ratio of $d+O(\sqrt[3]{d^2})$ for large $d$. We also present improved results for independent jobs and for jobs with special precedence constraints (e.g., series-parallel graphs and trees). Finally, we prove a lower bound of $d$ on the approximation ratio of any list scheduling scheme with local priority considerations. To the best of our knowledge, these are the first approximation results for moldable workflows with multiple resource requirements.

\end{abstract}

\section{Introduction}

Many complex scientific workflows that are running in today's High-Performance Computing (HPC) systems can be modeled as Directed Acyclic Graphs (DAGs), where the nodes represent the constituent jobs of the workflows and the edges represent the precedence constraints or dependencies among the jobs. While HPC systems often rely on dynamic runtime schedulers, such as KAAPI~\cite{KAAPI}, StarPU~\cite{StarPU} or PaRSEC~\cite{PaRSEC}, to ensure the efficient execution of these workflows, most existing schedulers focus only on the management of the computational resources (i.e., computing nodes or cores).
However, many of today's scientific applications need to process large amounts of data, and thus require not only the computational resources but also strong data management supports. Indeed, modern HPC systems are equipped with more levels of memory/storage (e.g., NVRAMs, SSDs, burst buffers~\cite{BurstBuffer}), as well as more advanced architecture and software features (e.g., high-bandwidth memory \cite{KNL}, cache partitioning \cite{vCAT},  bandwidth reservation \cite{MemGuard}) to facilitate efficient data transfer. All of these different types of resources could potentially be partitioned among the concurrently running jobs and thus exploited by the runtime schedulers to improve the overall application performance and system utilization.

In this paper, we study multi-resource scheduling for a computational workflow that is comprised of a set of parallel jobs with DAG-based precedence constraints. The goal is to simultaneously explore the availability of multiple types of resources by designing effective scheduling solutions that minimize the overall completion time, or \emph{makespan}, of the workflow. We focus on parallel jobs that are \emph{moldable} \cite{Feitelson97}, which allows the scheduler to select a variable set of resources for a job, but once the job starts execution, the resource allocations cannot be changed. In contrast to \emph{rigid} jobs, whose resource allocations are all static and hence fixed, moldable jobs can easily adapt to the different amounts of available resources, while in contrast to \emph{malleable} jobs, whose resource allocations can be dynamically varied during runtime, moldable jobs are much easier to design and implement. Given these advantages, moldable jobs have been offered by many computational kernels in scientific libraries. Moreover, the moldable job model is also amenable to the resource allocation patterns currently supported by many different resource types (e.g., computing cores, memory blocks, cache lines).

As the considered multi-resource scheduling problem contains the single-resource problem as a special case, it is known to be strongly NP-complete \cite{Du89_NPcomplete}. Thus, we focus on designing good approximation algorithms. In contrast to the single-resource problem, however, the multi-resource problem needs to consider the combined effect of multiple types of resources on the execution time of the jobs, which poses additional challenges to the scheduling problem. By adopting a two-phase approach \cite{Turek92} widely used for scheduling moldable jobs, we design a multi-resource, list-based scheduling algorithm. In particular, our algorithm first computes an approximate resource allocation for all jobs on different resource types, and then applies an extended list scheduling scheme to schedule the jobs. As list scheduling is easy to implement, the proposed algorithm can be readily applied to practical systems.

We prove the following main results for a system consisting of $d$ types of schedulable resources, under reasonable assumptions on the job execution times and speedups:
\begin{compactitem}
\item An approximation ratio of $1.619d+2.545\sqrt{d}+1$ for any $d$, and a ratio of $d+O(\sqrt[3]{d^2})$ for large $d$;
\item Improved approximations for some special graphs (e.g., series-parallel graphs, trees and independent jobs) with ratios of $1.619d + 1$ for any $d$ and $d+O(\sqrt{d})$ for large $d$.
\item A lower bound of $d$ on the approximation ratio of any list scheduling scheme with local priority considerations.
\end{compactitem}

To the best of our knowledge, these are the first approximation results for moldable workflows with multiple resource requirements.
They also improve upon the $2d$-approximation previously shown in \cite{Sun18_multiresource} for independent moldable jobs. The results demonstrate that our algorithm essentially achieves the optimal asymptotic approximation up to the dominating factor (i.e., $d$) among the generic class of local list scheduling schemes, thus matching the same asymptotic performance for rigid \cite{Garey75_multiresource} and malleable \cite{He07_multiresource} jobs.
Altogether, these results lay the theoretical foundation for multi-resource scheduling of parallel workflows.

The rest of this paper is organized as follows. Section \ref{sec.related} reviews some related work on moldable and multi-resource scheduling. Section \ref{sec.models} formally introduces the scheduling model and derives a lower bound on the optimal makespan. Section \ref{sec.algorithms} presents our multi-resource scheduling algorithm and analyzes its approximation ratios for general job graphs. Section \ref{sec.improved} proves improved results for some special graphs, including series-parallel graphs, trees and independent jobs. Section \ref{sec.listlowerbound} shows a lower bound on the performance of local list scheduling schemes, and finally, Section \ref{sec.conclusion} concludes the paper and briefly discusses open questions.

\section{Related Work}\label{sec.related}

This section reviews some related work on scheduling moldable parallel jobs, as well as on multi-resource scheduling under different job models and objectives.

\subsection{Moldable Job Scheduling}
Scheduling moldable parallel jobs to minimize the makespan is strongly NP-hard on $P \ge 5$ processors \cite{Du89_NPcomplete}, and the problem has been extensively studied in the literature from the perspective of approximation algorithms. Most prior work, however, has focused on a single type of resource while assuming different speedup models for the jobs.

For scheduling independent moldable jobs with arbitrary speedups, Turek et al.~\cite{Turek92} presented a 2-approximation list-based algorithm and a 3-approximation algorithm based on building shelves. Ludwig and Tiwari~\cite{Ludwig94} later improved the 2-approximation result with lower computational complexity.
For monotonic jobs, whose execution time $t(p)$ is non-decreasing in the number $p$ of allocated processors and whose work function $w(p) = p\cdot t(p)$ is non-decreasing in $p$, Mouni\'{e} et al.~\cite{Mounie07_dualapprox} presented a $(1.5+\epsilon)$-approximation algorithm using dual approximation. Jansen and Land \cite{Jansen18_PTAS} showed a lower complexity algorithm that achieves the same $(1.5+\epsilon)$-approximation as well as a PTAS, when the execution time functions of the jobs admit compact encodings.

For scheduling moldable jobs with precedence constraints, Lep\`{e}re et al. \cite{Lepere01_DAG} presented a $5.236$-approximation algorithm for monotonic jobs.
Jansen and Zhang \cite{Jansen06_DAG} improved the approximation ratio to around 4.73 for the same model, and recently, Chen \cite{Chen18_DAG} further improved it to around 3.42 using an iterative approximation method. Additionally, better approximation results have been obtained for jobs with special dependency graphs (e.g., series-parallel graphs and trees \cite{Lepere01_DAG, Lepere02_trees}) or special speedup models (e.g., concave speedup \cite{Jansen05_concave, Chen13_concave} and roofline speedup \cite{Wang92_DAG, Feldmann98_DAG}).

\subsection{Multi-Resource Scheduling}

Some approximation algorithms have been proposed on multi-resource scheduling to minimize makespan under different parallel job models.

Garey and Graham \cite{Garey75_multiresource} considered scheduling $n$ sequential jobs on $m$ identical machines with $d$ additional types of resources.
Further, each job has a fixed resource requirement from each resource type, making it essentially a \emph{rigid} job scheduling model. They presented a list-scheduling algorithm and proved three results: (1) an $m$-approximation for jobs with precedence constraints and when there is only one type of resource, i.e., $d=1$; (2) a $(d+1)$-approximation for independent jobs and when the number of machines is not a constraining factor, i.e., $m\ge n$; (3) a $(d+2-\frac{2d+1}{m})$-approximation for independent jobs with any $m\ge 2$. For the case of $d=1$, Demirci et al. \cite{Demirci18_resource} presented an improved $O(\log n)$-approximation for jobs with precedence constraints, and Niemeier and Wiese \cite{Niemeier12_resource} presented an improved $(2+\epsilon)$-approximation for independent jobs.

He et al.~\cite{He07_multiresource, He11_MQB} considered parallel jobs that are represented as direct acyclic graphs (DAGs) consisting of unit-size tasks, each of which requests a single type of resource from a total of $d$ resource types. Further, the amount of resources allocated to a job can be dynamically changed during runtime, making it essentially a \emph{malleable} job scheduling model. They showed that list scheduling achieves $(d+1)$-approximation for this model. Shmoys et al. \cite{Shmoys94_dagshop} considered a similar model while further restricting the tasks of each job to be processed sequentially. They called it the \emph{DAG-shop} scheduling model, and presented a polylog approximation result in number of machines and job length.

Sun et al.~\cite{Sun18_multiresource} considered scheduling independent \emph{moldable} jobs on $d$ types of resources. They presented a $2d$-approximation list-based algorithm and a $(2d+1)$-approximation shelf-based algorithm, thus generalizing the single-resource results in \cite{Turek92}. They also presented a technique to transform any $c$-approximation algorithm for a single resource type to a $cd$-approximation algorithm for $d$ types of resources. This work is the closest to ours, while we consider moldable jobs with precedence constraints. When jobs are independent, our main approximation result also improves the one in \cite{Sun18_multiresource} for a large number of resource types.

Beaumont et al.~\cite{Beaumont18_tworesource} and Eyraud-Dubois and Kumar \cite{Eyraud-Dubois20_tworesource} considered scheduling sequential jobs on two alternative types of resources (CPU and GPU) to minimize the makespan. In their model, each job can be chosen to execute on either resource type with different processing rates. They analyzed an approximation algorithm, called HeteroPrio, for both independent jobs and jobs with precedence constraints. The approximation ratios depend on the relative amount of resources in the two resource types. A recent survey on this \emph{alternative-resource} scheduling model can also be found in \cite{Beaumont20_survey}.

Additionally, some prior works have studied heuristic algorithms under various multi-resource scheduling models or objectives. Ghodsi et al.~\cite{Ghodsi11_DRF} focused on the objective of resource allocation fairness in a multi-user setting. They proposed the Dominant Resource Fairness (DRF) algorithm that aims at maximizing the minimum dominant share across all users.
Grandl et al.~\cite{Grandl_14_MPC} considered scheduling malleable jobs under four specific resource types (CPU, memory, disk and network). They designed a heuristic algorithm, called Tetris, that schedules jobs by considering the correlation between the job's peak resource demands and the machine's resource availabilities, with the goal of minimizing resource fragmentation.
NoroozOliaee et al.~\cite{NoroozOliaee14_Online} studied a similar problem but with two resources only (CPU and memory). They showed that a simple scheduling heuristic that uses Best Fit and Shortest Job First delivers good performance in terms of resource utilization and job queueing delays.

\section{Models}\label{sec.models}

This section presents the multi-resource scheduling model, gives a formal statement of the problem, and derives a lower bound on the optimal schedule.

\subsection{Scheduling Model}

We consider the problem of scheduling a set of $n$ moldable jobs on $d$ distinct types of resources (e.g., processor, memory, cache). Each resource type $i$ has a total amount $P^{(i)}$ of available resource. The jobs are \emph{moldable}, i.e., they can be executed using different amounts of resources from each resource type, but the resource usage cannot be changed once a job has started executing. For each job $j$, its execution time $\exectime{j}(\resources{j})$ depends on the \emph{resource allocation} $\resources{j} = (\resource{j}{1}, \resource{j}{2}, \cdots, \resource{j}{d})$, which specifies the amount of resource $\resource{j}{i} \ge 0$ allocated to the job for each resource type $i = 1, 2, \dots, d$.
We make the following reasonable assumptions on the resource allocation and execution time of the jobs.

\begin{assumption}[Integral Resources]
All resource allocations $\resource{j}{i}$'s for the jobs and the total amount of resources $\maxresource{i}$'s for all resource types are integers.
\end{assumption}

This is a natural assumption for discrete resources, such as processors. Other resource types, such as memory or cache, are typically allocated in discrete chunks as well (e.g., memory blocks, cache lines) in practical systems.

\begin{assumption}[Known Execution Times]
For each job $j$, its execution time function $\exectime{j}(\resources{j})$ is known for every possible resource allocation $p_j$.
\end{assumption}

In practice, the execution time function of an application could be obtained through one or more of the following approaches: application modeling or profiling, performance prediction or interpolation from historic data. Here, we are not concerned about how such a function is obtained.

\begin{assumption}[Monotonic Jobs]\label{assume.monotonic}
Given two resource allocations $\resources{j}$ and $q_{j}$ for a job $j$, we say that $\resources{j}$ is \emph{at most} $q_{j}$, denoted by $\resources{j} \preceq q_{j}$, if $\resource{j}{i} \le q_j^{(i)}$ for all $1\le i\le d$. The execution times of the job under these two allocations satisfy:
$$\exectime{j}(q_{j}) \le \exectime{j}(p_{j}) \leq \Big(\max_{i=1\dots d} q_{j}^{(i)}/p_{j}^{(i)}\Big) \cdot \exectime{j}(q_{j}) \ . $$
\end{assumption}

This generalizes the monotonic job assumption under a single resource type \cite{Lepere01_DAG, Mounie07_dualapprox}, which has been observed for many real-world applications. In particular, the first inequality specifies that the execution time of a job
is \emph{non-increasing} in the amount of resource allocated to the job\footnote{This assumption, however, is not restrictive, as we can discard any allocation that uses more resource than another allocation but results in a higher job execution time.}, and the second inequality restricts the job to have \emph{non-superlinear} speedup with respect to any resource type\footnote{Some parallel applications can achieve superlinear speedups with a combined effect of increased allocations in two or more resource types (e.g., the \emph{cache effect} \cite{Superlinear} when increasing both processor and cache allocations). We do not consider such superlinear speedup model in this paper.}. Note that we do not make any assumptions on a job $j$'s relative execution times under two resource allocations $\resources{j}$ and $q_{j}$ that are \emph{non-comparable}, i.e., $\resources{j} \npreceq q_{j}$ and $q_{j} \npreceq \resources{j}$.

Additionally, a set of \emph{precedence constraints} is specified for the jobs, which form a directed acyclic graph (DAG), $G = (V, E)$. Each node $j \in V$ in the graph represents a job and a directed edge $(j_1 \rightarrow j_2) \in E$ requires that job $j_2$ cannot start executing until the completion of job $j_1$. In this case, $j_1$ is called an immediate \emph{predecessor} of $j_2$, and $j_2$ is called an immediate \emph{successor} of $j_1$.

\subsection{Problem Statement}

The objective is to find a schedule for the jobs to minimize the maximum completion time, or the makespan. Specifically, a schedule is defined by the following two decisions:
\begin{itemize}
\item \emph{Resource allocation decision}: $\allresources = (\resources{1}, \resources{2}, \dots, \resources{n})$;
\item \emph{Starting time decision}: $\allstarttimes = (s_{1}, s_{2}, \dots, s_{n})$.
\end{itemize}

Given a pair of scheduling decisions $\allresources$ and $\allstarttimes$, the completion time of a job $j$ is defined as $c_j = s_j + \exectime{j}(\resources{j})$, and the makespan of the jobs is given by $T = \max_{j} c_j$. A schedule is \emph{valid} if it respects the following constraints:
\begin{itemize}
\item For each resource type $i$, the amount of resource utilized by all running jobs at any time does not exceed the total amount $\maxresource{i}$ of available resource;
\item If two jobs $j_1$ and $j_2$ have a precedence constraint, i.e., $j_1 \rightarrow j_2$, then the starting time of $j_2$ is no earlier than the completion time of $j_1$, i.e., $s_{j_2} \ge c_{j_1}$.
\end{itemize}

The above multi-resource scheduling problem is clearly NP-complete, as it contains the single-resource scheduling problem \cite{Jansen06_DAG, Lepere01_DAG} as a special case. Thus, we aim at designing approximation algorithms with bounded performance guarantees. An algorithm is said to be \emph{$r$-approximation} if its makespan satisfies $\frac{T}{T_{\opt}} \le r$ for any set of jobs, where $T_{\opt}$ denotes the optimal makespan.

\subsection{Lower Bound on Optimal Makespan}

We now derive a lower bound on the optimal makespan. To that end, we define the following concepts given a resource allocation decision $\allresources = (\resources{1}, \resources{1}, \dots, \resources{n})$ for the jobs.

\begin{definition}
For each job $j$:
\begin{itemize}
\item $w_j^{(i)}(\resources{j}) \!=\! p_j^{(i)}\cdot \exectime{j}(\resources{j})$\emph{:} work on resource type $i$;
\item $a_j^{(i)}(p_j) \!=\! \frac{w_j^{(i)}(\resources{j})}{P^{(i)}}$\emph{:} area (or normalized work) on resource type $i$;
\item $a_j(p_j) \!=\! \frac{1}{d}\sum_{i=1}^{d} a_j^{(i)}(p_j)$\emph{:} average area over all resource types.
\end{itemize}
\end{definition}

\begin{definition}
For the set of jobs:
\begin{itemize}
\item $W^{(i)}(\allresources) \!=\! \sum_{j=1}^{n} w_j^{(i)}(\resources{j})$\emph{:} total work on resource type $i$;
\item $A^{(i)}(\allresources) \!=\! \frac{W^{(i)}(\allresources)}{P^{(i)}} = \sum_{j=1}^{n} a_j^{(i)}(p_j)$\emph{:} total area on resource type $i$;
\item $A(\allresources) \!=\! \frac{1}{d}\sum_{i=1}^{d} A^{(i)}(\allresources) \!=\! \sum_{j=1}^{n} a_j(p_j)$\emph{:} average total area over all resource types;
\item $C(\allresources, f) \!=\! \sum_{j \in f} \exectime{j}(\resources{j})$\emph{:} total execution time of all the jobs along a particular path $f$ in the graph\footnote{A path is a sequence of jobs with linear precedence, i.e., $f = (j_{\pi(1)} \rightarrow j_{\pi(2)} \rightarrow \dots \rightarrow j_{\pi(v)})$, where the first job $j_{\pi(1)}$ does not have any predecessor in the graph and the last job $j_{\pi(v)}$ does not have any successor. };
\item $C(\allresources) \!=\! \max_{f} C(\allresources, f)$\emph{:} critical path length, i.e., total execution time of the jobs along a critical (longest) path in the graph;
\item $L(\allresources)=\max(A(\allresources), C(\allresources))$\emph{:} maximum of average total area $A(\allresources)$ and critical path length $C(\allresources)$.
\end{itemize}
\end{definition}

We further define $L_{\min} = \min_{\allresources}L(\allresources)$ to be the minimum value of $L(\allresources)$ among all possible resource allocations, and let $\allresources^*$ denote a resource allocation such that $L(\allresources^*) = L_{\min}$. The following lemma shows that $L_{\min}$ serves as a lower bound on the optimal makespan.

\begin{lemma}\label{lem.makespan_lb}
$T_{\opt} \ge L_{\min}$.
%\end{align}
\end{lemma}

\begin{proof}
We first show that, given any resource allocation $\allresources$, the makespan produced by any schedule must satisfy $T \ge \max(A(\allresources), C(\allresources))$. The bound $T \ge C(\allresources)$ is trivial, since the jobs along the critical path must be executed sequentially, so the makespan is at least $C(\allresources)$. To derive the bound $T \ge A(\allresources)$, we observe that the average total area $A(\allresources)$ in any valid schedule with makespan $T$ must satisfy:
\begin{align*}
A(\allresources) &= \frac{1}{d}\sum_{i=1}^{d}\sum_{j=1}^{n} \frac{w_j^{(i)}(\resources{j})}{\maxresource{i}} \\
&= \frac{1}{d}\sum_{i=1}^{d} \frac{1}{\maxresource{i}} \sum_{j=1}^{n} w_j^{(i)}(\resources{j}) \\
&\le \frac{1}{d}\sum_{i=1}^{d} \frac{1}{\maxresource{i}} \cdot (\maxresource{i} \cdot T) = T \ .
%&= T \ .
\end{align*}
The inequality $\sum_{j=1}^{n} w_j^{(i)}(\resources{j}) \le \maxresource{i} \cdot T$ is because $\maxresource{i} \cdot T$ is the maximum amount of work that can be allocated to the jobs within time $T$ on any resource type $i$ with total amount of resource $P^{(i)}$.

Suppose the optimal schedule uses a resource allocation $\allresources_{\opt}$. Then, its makespan must satisfy:
$$T_{\opt} \ge \max\big(A(\allresources_{\opt}), C(\allresources_{\opt}) \big) = L(\allresources_{\opt}) \ge L(\allresources^*) = L_{\min}.$$
The last inequality is because $L(\allresources^*)$ is the minimum $L(\allresources)$ among all possible resource allocations, including $\allresources_{\opt}$.
\end{proof}

\section{A Multi-Resource Scheduling Algorithm and Approximation Results}\label{sec.algorithms}

In this section, we present a multi-resource scheduling algorithm and analyze its approximation ratio for general DAGs. The algorithm adopts the \emph{two-phase approach} that has been widely used for scheduling moldable jobs on a single type of resource \cite{Turek92, Lepere01_DAG, Jansen06_DAG}.
\subsection{Phase 1: Resource Allocation}\label{sec.allocation}

\subsubsection{Discrete Time-Cost Tradeoff (DTCT) Problem} To allocate resources for the jobs, we consider a relevant discrete time-cost tradeoff problem \cite{De95_tradeoff}, which has been studied in the literature of operations research and project management.

\begin{definition}[Discrete Time-Cost Tradeoff (DTCT)]\label{def.DTCT}
Suppose a project consists of $n$ precedence-constrained tasks. Each task $j$ can be executed using several different alternatives and each alternative $i$ takes time $t_{j, i}$ and has cost $c_{j, i}$. Further, for any two alternatives $i_1$ and $i_2$, if $i_1$ is faster than $i_2$, then $i_1$ is more costly than $i_2$, i.e.,
\begin{align}\label{eq.time_cost}
t_{j, i_1} \le t_{j, i_2} \Rightarrow c_{j, i_1} \ge c_{j, i_2} \ .
\end{align}
Given a project realization $\sigma$ that specifies which alternative is chosen for each task, the total project duration $D(\sigma)$ is defined as the sum of times of the tasks along the critical path, and the total cost $B(\sigma)$ is defined as the sum of costs of all tasks. The objective is to find a realization $\sigma^*$ that minimizes the total project duration $D(\sigma^*)$ and the total cost $B(\sigma^*)$.
\end{definition}

The above DTCT problem is obviously bicriteria, and a tradeoff exists between the total project duration and the total cost. Two problem variants have been commonly studied, both of which are shown to be NP-complete \cite{De97_tradeoffcomplexity}:
\begin{itemize}
\item \emph{Budget Problem}: Given a total cost budget $B$, minimize the project duration $D(\sigma)$ subject to $B(\sigma)\le B$;
\item \emph{Deadline Problem}: Given a project deadline $D$, minimize the total cost $B(\sigma)$ subject to $D(\sigma)\le D$.
\end{itemize}

For both problems, Skutella \cite{Skutella98_tradeoff} presented a polynomial-time algorithm, which, given any feasible budget-deadline pair $(B, D)$, finds a realization $\sigma$ for the project that satisfies: $D(\sigma) \le \frac{D}{\rho}$ and $B(\sigma) \le \frac{B}{1-\rho}$, for any $\rho \in (0, 1)$.\footnote{In essence, this bicriteria approximation algorithm first transforms each task of the project into a set of virtual tasks, and then constructs a relaxed linear program (LP) for the transformed problem. The relaxed LP either minimizes $D(\sigma)$ subject $B(\sigma) \le B$ or minimizes $B(\sigma)$ subject $D(\sigma) \le D$. In either case, the result can be obtained by rounding the optimal fractional solution to the relaxed LP based on the parameter $\rho$. }
\smallskip

\subsubsection{Allocating Resources to Jobs} We transform our resource allocation problem to the DTCT problem and solve it using the approximation result in \cite{Skutella98_tradeoff}. To that end, a task $j$ is created for each job $j$ in the graph, with the set of alternatives for the task corresponding to the set of resource allocations for the job. The execution time $t_{j, i}$ of task $j$ with alternative $i$ is then defined as the execution time $\exectime{j}(\resources{j})$ of job $j$ with the corresponding resource allocation $\resources{j}$, and the cost $c_{j, i}$ is defined as the average area $a_j(\resources{j})$.

Let $\mathcal{S}$ denote the set of all $Q = \prod_{i=1}^{d} P^{(i)}$ possible resource allocations for a job. To ensure that Condition (\ref{eq.time_cost}) in Definition \ref{def.DTCT} is satisfied, we discard, for each job $j$, the subset $\mathcal{D}_j \subset \mathcal{S}$ of \emph{dominated} allocations, which is defined as:
\begin{align}\label{eq.dominated}
\mathcal{D}_j \!=\! \{\resources{j} \mid \exists q_{j}, \exectime{j}(q_{j}) < \exectime{j}(\resources{j}) \text{ and } a_j(q_{j}) < a_j(\resources{j})\} \ ,
\end{align}
and only use the remaining set of \emph{non-dominated} allocations, denoted by $\mathcal{N}_j = \mathcal{S}\backslash \mathcal{D}_j$, to create the alternatives of the task.
Thus, a realization $\sigma$ for the project corresponds to a resource allocation decision $\allresources$ for the jobs. The total project duration $D(\sigma)$ then corresponds to the total execution time $C(\allresources)$ of the jobs, and the total cost $B(\sigma)$ corresponds to the average total area $A(\allresources)$.

A resource allocation decision $\allresources = (\resources{1}, \resources{2}, \dots, \resources{n})$ is said to be \emph{non-dominated} if the allocation for every job is non-dominated, i.e., $\resources{j} \in \mathcal{N}_j$ for all $j=1,\dots,n$. The following lemma shows that the minimum makespan lower bound $L_{\min}$ can be achieved by a non-dominated resource allocation.

\begin{lemma}\label{lem.non-dominate}
There exists a non-dominated resource allocation $\allresources^{*} = (\resources{1}^*, \resources{2}^*, \dots, \resources{n}^*)$ that achieves
$L(\allresources^*) = L_{\min}$.
\end{lemma}

\begin{proof}
Consider any resource allocation $\mathbf{q}^*=(q^*_{1},q^*_{2},\dots,q^*_{n})$ that achieves $L(\mathbf{q}^*) = L_{\min}$, and suppose it contains a dominated allocation $q^*_{j}\in \mathcal{D}_j$ for a job $j$. Then, by replacing $q^*_{j}$ with a non-dominated allocation $q'^{*}_{j} \in \mathcal{N}_j$ that dominates $q^*_{j}$, i.e., $\exectime{j}(q'^{*}_j) < \exectime{j}(q^*_{j})$ and $a_j(q'^{*}_j) < a_j(q^*_{j})$, we get a new resource allocation $\mathbf{q'}^*=(q^*_{1},\dots,q^*_{j-1},q'^{*}_j,q^*_{j+1},\dots,q^*_{n})$, which satisfies $A(\mathbf{q'}^*) < A(\mathbf{q}^*)$ and $C(\mathbf{q'}^*) \le C(\mathbf{q}^*)$. This implies $L(\mathbf{q'}^*) \le L(\mathbf{q}^*) = L_{\min}$. Repeating the process above for every job with a dominated allocation results in an overall non-dominated allocation $\mathbf{p}^*$ and proves the lemma.
\end{proof}

We can now find a resource allocation $\mathbf{p'}$ for the jobs (or equivalently a realization $\sigma'$ in the corresponding DTCT problem), with the following property.
\begin{lemma}\label{lem.allocation}
For any $\rho \in (0, 1)$, a resource allocation $\mathbf{p'} = (p'_1, p'_2,\dots,p'_n)$ can be found in polynomial time that satisfies:
\begin{align}
C(\mathbf{p'}) &\le \frac{T_{\opt}}{\rho} \ , \label{eq.Cbound}\\
A(\mathbf{p'}) &\le \frac{T_{\opt}}{1-\rho} \  . \label{eq.Abound}
\vspace{-0.2in}
\end{align}
\end{lemma}

\begin{proofsketch}
The result can be obtained by adapting the algorithm in \cite{Skutella98_tradeoff}, which minimizes the project duration (or total cost) subject to a known budget $B$ (or deadline $D$) for the DTCT problem. Without knowing the value of this constraint a priori, we can still achieve the same approximations by adopting the technique used in \cite{Jansen06_DAG} for the problem with a single resource type.
Specifically, the relaxed LP originally formulated in \cite{Skutella98_tradeoff} can be modified and applied to our problem as follows: minimize the lower bound $L(\mathbf{p})$ instead, subject to two additional constraints $C(\mathbf{p}) \le L(\mathbf{p})$ and $A(\mathbf{p}) \le L(\mathbf{p})$. Then, by rounding the optimal fractional solution $\mathbf{\bar{p}}^*$ to this modified LP, we can get a resource allocation $\mathbf{p}'$ that satisfies: $C(\mathbf{p'}) \le \frac{C(\mathbf{\bar{p}}^*)}{\rho} \le \frac{L(\mathbf{\bar{p}}^*)}{\rho}$ and $A(\mathbf{p'}) \le \frac{A(\mathbf{\bar{p}}^*)}{1-\rho} \le \frac{L(\mathbf{\bar{p}}^*)}{1-\rho}$. Since the optimal fractional solution $\mathbf{\bar{p}}^*$ must result in an objective not greater than the one achieved by any (non-dominated) integral solution $\mathbf{p}^*$, and based on Lemma \ref{lem.non-dominate}, we have $L(\mathbf{\bar{p}}^*) \le L(\mathbf{p}^*) = L_{\min}$.
The result then directly follows by applying the makespan lower bound in Lemma \ref{lem.makespan_lb}.
\end{proofsketch}

\subsubsection{Adjusting Resource Allocation}
Lastly, we adjust the resource allocation $\mathbf{p'}$ (obtained above with a value of $\rho$ to be determined later) to get the final resource allocation $\mathbf{p}$ for the jobs. The aim is to limit the maximum resource utilization of any job under any resource type, thus facilitating more efficient list scheduling (see Section \ref{sec.scheduling}).
As with the case for a single type of resource \cite{Lepere01_DAG, Jansen06_DAG}, we choose a parameter $\mu \in (0, 0.5)$, whose value will also be determined later, and define the resource allocation for each job $j$ on each resource type $i$ as follows:
\begin{align}\label{eq.adjust}
p_j^{(i)} = \begin{cases}
\lceil \mu P^{(i)} \rceil,  & \text{~if~} {p'}_j^{(i)} > \lceil \mu P^{(i)} \rceil \\
{p'}_j^{(i)}, & \text{~otherwise~}
\end{cases}
\end{align}
where ${p'}_j^{(i)}$ is the corresponding resource allocation in $\mathbf{p'}$. The $p_j^{(i)}$'s will then form the final resource allocation $\mathbf{p}$.

A job $j$ is said to be \emph{adjusted} if its final resource allocation $p_j$ is reduced from the initial allocation $p'_j$ in any resource type; otherwise, the job is said to be \emph{unadjusted}. The following lemma shows the properties of any adjusted job.
\begin{lemma}\label{lem.adjust}
For any adjusted job $j$, its execution time satisfies:
\begin{align}\label{eq.tj}
t_j(p_j) \le \frac{t_j(p'_j)}{\mu} \ ,
\end{align}
and its area on any resource type $i$ is bounded by:
\begin{align}\label{eq.aj}
a_j^{(i)}(p_j) \le d\cdot a_j(p'_j) \ ,
\end{align}
if the total amount of resource type $i$ satisfies $P^{(i)} \ge \frac{1}{\mu^2}$.
\end{lemma}

\begin{proof}
For any adjusted job $j$, let $x_j^{(i)} = \frac{p'^{(i)}_j}{p^{(i)}_j}$ denotes its resource reduction factor on any resource type $i$, and let $k = \argmin_{i=1\dots d}x_j^{(i)}$ denote the resource type with the largest reduction factor for $j$.

Since the job's final resource allocation $p_j$ is at most its initial allocation $p'_j$, i.e., $p_j \preceq p'_j$, and according to the adjustment procedure in Equation (\ref{eq.adjust}), we have $x_j^{(k)} \le \frac{P^{(k)}}{\lceil \mu P^{(k)} \rceil} \le \frac{1}{\mu}$. Thus, based on Assumption \ref{assume.monotonic}, we can get $t_j(p_j) \le \big(\max_{i=1\dots d} x_j^{(i)}\big) \cdot t_j(p'_j) = x_j^{(k)} \cdot t_j(p'_j) \le \frac{t_j(p'_j)}{\mu}$.

To prove the area bound, we distinguish three cases.

Case (1): For resource type $k$ with the largest reduction factor, we have $w_j^{(k)}(p_j) = p_j^{(k)}\cdot t_j(p_j) \le \frac{p'^{(k)}_j}{x^{(k)}_j}\cdot (x^{(k)}_j \cdot t_j(p'_j)) = p'^{(k)}_j \cdot t_j(p'_j) = w_j^{(k)}(p'_j)$. Thus, the area of the job on resource type $k$ satisfies $a_j^{(k)}(p_j) = \frac{w_j^{(k)}(p_j)}{P^{(k)}} \le \frac{w_j^{(k)}(p'_j)}{P^{(k)}} \le \sum_{\ell=1}^{d} \frac{w_j^{(\ell)}(p'_j)}{P^{(\ell)}} = d\cdot a_j(p'_j)$.

Case (2): For any resource type $i \neq k$ with $p_j^{(i)} \le \lfloor \mu P^{(i)} \rfloor \le \mu P^{(i)}$, and since $p_j^{(k)} = \lceil \mu P^{(k)} \rceil \ge \mu P^{(k)}$, we have $a_j^{(i)}(p_j) = \frac{w_j^{(i)}(p_j)}{P^{(i)}} = \frac{p_j^{(i)}\cdot t_j(p_j)}{P^{(i)}} \le \frac{\mu P^{(i)} \cdot t_j(p_j)}{P^{(i)}} \le \mu \cdot x^{(k)}_j \cdot t_j(p'_j) = \mu \cdot \frac{p'^{(k)}_j \cdot t_j(p'_j)}{p_j^{(k)}} \le \mu \cdot~\frac{w_j^{(k)}(p'_j)}{\mu P^{(k)}} = \frac{w_j^{(k)}(p'_j)}{P^{(k)}} \le \sum_{\ell=1}^{d} \frac{w_j^{(\ell)}(p'_j)}{P^{(\ell)}} = d\cdot a_j(p'_j)$.

Case (3): For any resource type $i \neq k$ with $p_j^{(i)} = \lceil \mu P^{(i)} \rceil \le \mu P^{(i)} + 1$, by following the derivation steps in Case (2), we can get $a_j^{(i)}(p_j) \le \left(1+\frac{1}{\mu P^{(i)}}\right)\frac{w_j^{(k)}(p'_j)}{P^{(k)}} \le \sum_{\ell=1}^{d} \frac{w_j^{(\ell)}(p'_j)}{P^{(\ell)}} + \frac{w_j^{(k)}(p'_j)}{\mu P^{(i)}P^{(k)}} - \frac{w_j^{(i)}(p'_j)}{P^{(i)}} = \sum_{\ell=1}^{d} \frac{w_j^{(\ell)}(p'_j)}{P^{(\ell)}} + \frac{t_j(p'_j)}{P^{(i)}} \left(\frac{p'^{(k)}_j}{\mu P^{(k)}} - p'^{(i)}_j \right)$. Since $p'^{(k)}_j \le P^{(k)}$ and $p'^{(i)}_j \ge \lceil \mu P^{(i)} \rceil \ge \mu P^{(i)}$, we have $\frac{p'^{(k)}_j}{\mu P^{(k)}} - p'^{(i)}_j \le \frac{1}{\mu} - \mu P^{(i)}$, which is at most 0 when $P^{(i)} \ge \frac{1}{\mu^2}$. In this case, we get $a_j^{(i)}(p_j) \le \sum_{\ell=1}^{d} \frac{w_j^{(\ell)}(p'_j)}{P^{(\ell)}} = d\cdot a_j(p'_j)$.
\end{proof}

Algorithm \ref{alg.allocation} summarizes all three steps involved in this first phase of the multi-resource scheduling algorithm.

{\LinesNumberedHidden
\begin{algorithm}[t]
\small
\caption{Resource Allocation (Phase 1)}\label{alg.allocation}
\KwIn{For each job $j$, the execution time $t_j(p_j)$ and the average normalized work $a_j(p_j)$ under all possible resource allocations, given values for the parameters $\rho$ and $\mu$. }
\KwOut{Resource allocation decision $\mathbf{p} \!=\! (p_1, p_2, \dots, p_n)$ for all jobs. }
\Begin{
	\textbf{(Step 1)}: For each job $j$, discard the subset $\mathcal{D}_j \subset \mathcal{S}$ of dominated resource allocations as defined in Equation (\ref{eq.dominated})\;
    \textbf{(Step 2)}: Transform the resource allocation problem to the DTCT problem and adapt the algorithm in \cite{Skutella98_tradeoff} to obtain an initial allocation decision $\mathbf{p'}$ that satisfies Equations (\ref{eq.Cbound}) and (\ref{eq.Abound})\;
    \textbf{(Step 3)}: For each job $j$ and each resource type $i$, adjust the initial allocation in $\mathbf{p'}$ based on Equation (\ref{eq.adjust}) to obtain a final resource allocation decision $\mathbf{p}$ that satisfies Equations (\ref{eq.tj}) and (\ref{eq.aj}).
	}
\end{algorithm}
}

\subsection{Phase 2: List Scheduling}\label{sec.scheduling}

\subsubsection{Algorithm Description} The second phase schedules the jobs by making a starting time decision $\allstarttimes$, given the resource allocation decision $\allresources$ determined by the first phase. This is done through a modified list scheduling strategy, as shown in Algorithm~\ref{alg.list}, that extends to multiple types of resources.

A job is said to be \emph{ready} if all of its immediate predecessors in the precedence graph have been completed or if the job has no immediate predecessor. The algorithm starts by inserting all ready jobs into a queue $\mathcal{Q}$. Then, at time 0 or whenever a running job $k$ completes and hence releases resources, the algorithm inserts, into the queue $\mathcal{Q}$, any new job $k'$ that becomes ready due to the completion of job $k$.
It then goes through the list of all ready jobs in $\mathcal{Q}$ and schedules each job $j$ that can be executed at the current time if its resource allocation $p_j$ can be met by the amount of available resources in all resource types.

We point out that the ready jobs can be inserted into the queue in any order without affecting the approximation ratio of the algorithm. In practice, giving priority to certain jobs (e.g., with longer execution time or on the critical path) may yield better performance.

{\LinesNumberedHidden
\begin{algorithm}[t]
    \small
	\caption{List Scheduling (Phase 2)}\label{alg.list}
	\KwIn{Resource allocation decision $\mathbf{p} \!=\! (p_1, p_2, \dots, p_n)$ for all jobs, and their precedence constraints. }
    \KwOut{A list schedule for the jobs with starting time decision $\allstarttimes \!=\! (s_{1}, s_{2}, \dots, s_{n})$.}
	\Begin{
		insert all ready jobs into a queue $\mathcal{Q}$\; %Organize all jobs in a list $L$ according to some priority rule\;
        $P_{avail}^{(i)} \leftarrow P^{(i)}, \forall i$\;
        \When{\textnormal{at time $0$ or a job $k$ completes execution}}{
            $curr\_time \leftarrow getCurrentTime()$\;
            $P_{avail}^{(i)} \leftarrow P_{avail}^{(i)} + p_{k}^{(i)}, \forall i$\;
            \For{\textnormal{each job $k'$ that becomes ready}}{
                insert job $k'$ into queue $\mathcal{Q}$\;
            }
            \For{\textnormal{each job }$j \in \mathcal{Q}$}{
                \If{$P_{avail}^{(i)} \ge p_j^{(i)}, \forall i$}{
                    $s_j \leftarrow curr\_time$ and execute job $j$ now\;
                    $P_{avail}^{(i)} \leftarrow P_{avail}^{(i)} - p_{j}^{(i)}, \forall i$\;
                    remove job $j$ from queue $\mathcal{Q}$\;
                }
            }
		}
	}
\end{algorithm}
}

\subsubsection{Properties of List Scheduling}
We now derive some properties of the list scheduling algorithm, which will be used later in the analysis of the overall multi-resource scheduling algorithm.

We first define some notations. Let $T$ denote the makespan of a list schedule. We note that the algorithm only allocates and de-allocates resources upon job completions. Hence, the entire schedule's duration $[0, T]$ can be partitioned into a set $\mathcal{I} = \{I_1, I_2, \dots\}$ of non-overlapping intervals, where jobs only start (or complete) at the beginning (or end) of an interval, and the amount of utilized resource for any resource type does not change during an interval. For any resource type $i$, let $P^{(i)}_{util}(I)$ denote the total amount of utilized resources from all jobs that are running during interval $I \in \mathcal{I}$. We further classify the set of intervals into the following three categories.
\begin{compactitem}
\item $\mathcal{I}_1$: set of intervals during which the amount of utilized resources is at most $\lceil \mu P^{(i)} \rceil - 1$ for all resource type $i$, i.e., $\mathcal{I}_1 = \{I \mid \forall i, P^{(i)}_{util}(I) \le \lceil \mu P^{(i)} \rceil - 1 \}$.
\item $\mathcal{I}_2$: set of intervals during which there exists a resource type $k$ that utilizes at least $\lceil \mu P^{(k)} \rceil$ amount of resources, but the amount of utilized resources is at most $\lceil (1-\mu) P^{(i)} \rceil - 1$ for all resource type $i$, i.e., $\mathcal{I}_2 = \{I \mid \exists k, P^{(k)}_{util}(I) \ge \lceil \mu P^{(k)} \rceil \text{ and }  \forall i, \\ P^{(i)}_{util}(I) \le \lceil (1-\mu) P^{(i)} \rceil - 1 \}$.
\item $\mathcal{I}_3$: set of intervals during which there exists a resource type $k$ that utilizes at least $\lceil (1-\mu) P^{(k)} \rceil$ amount of resources, i.e., $\mathcal{I}_3 = \{I \mid \exists k, P^{(k)}_{util}(I) \ge \lceil (1-\mu) P^{(k)} \rceil \}$.
\end{compactitem}

Let $|I|$ denote the duration of an interval $I$, and let $T_1 = \sum_{I\in \mathcal{I}_1} |I|$, $T_2 = \sum_{I\in \mathcal{I}_2} |I|$ and $T_3 = \sum_{I\in \mathcal{I}_3} |I|$ be the total durations of the three categories of intervals, respectively. Since $\mathcal{I}_1$, $\mathcal{I}_2$ and $\mathcal{I}_3$ are obviously disjoint and partition $\mathcal{I}$, we have:
\begin{align}\label{eq.T}
T = T_1 + T_2 + T_3 \ .
\end{align}

Furthermore, for each job $j$ and each interval $I$, we define $\beta_{j, I}$ to be the \emph{fraction} of the job executed during that interval. For instance, if one third of job $j$ is executed in interval $I$ and two thirds of the job is executed in interval $I'$, we have $\beta_{j, I} = 1/3$ and $\beta_{j, I'} = 2/3$. Note that the fraction is defined in terms of either the execution time or the area (work) of the job, which are equivalent here since the resource allocation of the job has been fixed. Thus, for each job $j$, we have $\sum_{I\in \mathcal{I}} \beta_{j, I} = 1$.

The following lemma bounds the durations of the first two categories of intervals in terms of the execution time along the critical path of the initial resource allocation $\mathbf{p'}$. % (from Step 2 of Algorithm \ref{alg.allocation}).
\begin{lemma}[Critical-Path Bound]\label{lem.cp_bound}
For any choice of $\mu \in (0, 0.5)$, we have $T_1 + \mu T_2 \le C(\mathbf{p'})$.
\end{lemma}

\begin{proof}

For any interval $I \in \mathcal{I}_1 \cup \mathcal{I}_2$, the amount of utilized resource for any resource type $i$ is at most $\lceil (1-\mu) P^{(i)} \rceil - 1$, so the amount of available resource is at least $P^{(i)} + 1 - \lceil (1-\mu) P^{(i)} \rceil \ge \lceil \mu P^{(i)} \rceil$. According to the resource allocation algorithm, any job is allocated at most $\lceil \mu P^{(i)} \rceil$ amount of resource for resource type $i$. Thus, there is sufficient resource available to execute any additional job (if one is ready) during any interval $I \in \mathcal{I}_1 \cup \mathcal{I}_2$. This implies that there is no ready job in the queue $\mathcal{Q}$, since otherwise the list scheduling algorithm would have scheduled the job.

In list scheduling, it is known that there exists a path $f$ in the graph such that whenever there is no ready job in the queue, some job along that path is running \cite{Feldmann98_DAG, Lepere01_DAG, Jansen06_DAG}. Thus, during any interval $I \in \mathcal{I}_1 \cup \mathcal{I}_2$, some job along path $f$ is running, and we let $j(I)\in f$ denote such a job.

Now, consider the initial resource allocation $\mathbf{p'}$.
%and the total execution time of all jobs along path $f$ $f'$ in the graph given $\mathbf{p'}$, i.e., $C(\mathbf{p'}, f') = C(\mathbf{p'})$.
During any interval $I \in \mathcal{I}_1$, the amount of utilized resource for any resource type $i$ is at most $\lceil \mu P^{(i)} \rceil - 1$, so job $j(I)$ must be unadjusted. Thus, we have $t_{j(I)}(p_{j(I)}) = t_{j(I)}(p'_{j(I)})$. However, during any interval $I \in \mathcal{I}_2$, job $j(I)$ could be adjusted, and thus, according to Lemma \ref{lem.adjust} (Inequality (\ref{eq.tj})), we have $\mu \cdot t_{j(I)}(p_{j(I)}) \le t_{j(I)}(p'_{j(I)})$. We can then derive:
\vspace{-0.07in}
\begin{align*}
T_1 + \mu T_2 &= \sum_{I\in \mathcal{I}_1} t_{j(I)}(p_{j(I)}) \cdot \beta_{j(I), I} + \mu \sum_{I\in \mathcal{I}_2} t_{j(I)}(p_{j(I)}) \cdot \beta_{j(I), I} \\
&\le \sum_{I\in \mathcal{I}_1} t_{j(I)}(p'_{j(I)}) \cdot \beta_{j(I), I} + \sum_{I\in \mathcal{I}_2} t_{j(I)}(p'_{j(I)}) \cdot \beta_{j(I), I} \\
&\le \sum_{j \in f} \Big(t_{j}(p'_{j}) \cdot \!\!\sum_{I\in \mathcal{I}_1 \cup \mathcal{I}_2} \beta_{j, I} \Big) \\
&\le \sum_{j \in f} t_{j}(p'_{j}) = C(\allresources', f) \le C(\allresources') \ . \qedhere
\end{align*}
\end{proof}

The following lemma bounds the durations of the last two categories of intervals in terms of the average total area of the initial resource allocation $\mathbf{p'}$.
\begin{lemma}[Area Bound]\label{lem.area_bound}
For any choice of $\mu \in (0, 0.5)$, if $P^{\min} = \min_{i} P^{(i)} \ge \frac{1}{\mu^2}$, we have $\mu T_2 + (1-\mu) T_3 \le d\cdot A(\mathbf{p'})$.
\end{lemma}

\begin{proof}
For any interval $I \in \mathcal{I}_2$, there exists a resource type $i$ such that the amount of utilized resource is at least $\lceil \mu P^{(i)} \rceil$ based on the definition of $\mathcal{I}_2$. Therefore, the total work done on resource type $i$ from all jobs during this interval satisfies: $\sum_{j=1}^{n} \beta_{j, I} \cdot w_j^{(i)}(p_j) \ge |I| \cdot \lceil \mu P^{(i)} \rceil \ge |I| \cdot \mu P^{(i)}$. Thus, we have: $\mu \cdot |I| \le \sum_{j=1}^{n} \beta_{j, I} \cdot \frac{w_j^{(i)}(p_j)}{P^{(i)}} = \sum_{j=1}^{n} \beta_{j, I} \cdot a_j^{(i)}(p_j) \le d \sum_{j=1}^{n} \beta_{j, I} \cdot a_j(p'_j)$.
The last inequality is due to Lemma \ref{lem.adjust} (Inequality (\ref{eq.aj})), if $P^{(i)}\ge \frac{1}{\mu^2}$.
Note that Inequality (\ref{eq.aj}) was proven for any adjusted job but it obviously holds for unadjusted jobs as well. Thus, if $P^{\min} = \min_{i=1\dots d} P^{(i)} \ge \frac{1}{\mu^2}$, we can derive:
\begin{align}\label{eq.T2}
\mu T_2 &= \mu \sum_{I\in \mathcal{I}_2} |I| \nonumber \\
&\le d \sum_{I\in \mathcal{I}_2} \sum_{j=1}^{n} \beta_{j, I} \cdot a_j(p'_j) \nonumber \\
&= d \sum_{j=1}^{n} \Big(a_j(p'_j) \cdot \sum_{I\in \mathcal{I}_2} \beta_{j, I} \Big) \ .
\end{align}

For any interval $I \in \mathcal{I}_3$, there exists a resource type $i$ such that the amount of utilized resource is at least $\lceil (1-\mu) P^{(i)} \rceil$. Using the same argument, we can derive:
\begin{align}\label{eq.T3}
(1-\mu) T_3 \le d \sum_{j=1}^{n} \Big(a_j(p'_j) \cdot \sum_{I\in \mathcal{I}_3} \beta_{j, I} \Big) \ .
\end{align}

Thus, combining Inequalities (\ref{eq.T2}) and (\ref{eq.T3}), we can get:
\begin{align*}
\mu T_2 + (1-\mu) T_3 &\le d \sum_{j=1}^{n} \Big(a_j(p'_j) \cdot \!\!\sum_{I\in \mathcal{I}_2 \cup \mathcal{I}_3} \beta_{j, I} \Big) \\
&\le d \sum_{j=1}^{n} a_j(p'_j) =  d \cdot A(\mathbf{p'}) \ .  \qedhere
\end{align*}
\end{proof}

\subsection{Approximation Results}\label{sec.analysis}

We now derive the main approximation results of the multi-resource scheduling algorithm, which combines the resource allocation phase (Algorithm \ref{alg.allocation}) and the list scheduling phase (Algorithm \ref{alg.list}).
The following theorem shows its approximation ratio for any number $d$ of resource types.
\begin{theorem}\label{thm.main}
For any $d\ge 1$ and if $P^{\min} \ge 7$, the performance of the multi-resource scheduling algorithm satisfies:
\begin{align*}
\frac{T}{T_{\opt}} \le \phi d +2\sqrt{\phi d} +1 \le 1.619d+2.545\sqrt{d}+1 \ ,
\end{align*}
where $\phi = \frac{1+\sqrt{5}}{2}$ is the golden ratio. The result is achieved at $\mu^* = 1-\frac{1}{\phi}\approx 0.382$ and $\rho^* = \frac{1}{\sqrt{\phi d}+1}\approx \frac{1}{1.272\sqrt{d}+1}$.
\end{theorem}

We point out that $P^{\min} \ge 7$ represents a reasonable condition on the total amount of most discrete resource types (e.g., processors, memory blocks, cache lines).

\begin{proof}
Based on the analysis of the list scheduling algorithm, by substituting $T_1$ from Lemma (\ref{lem.cp_bound}) and $T_3$ from Lemma (\ref{lem.area_bound}) into $T = T_1 + T_2 + T_3$, and if $P^{\min} \ge \frac{1}{\mu^2}$, we get:
\begin{align*}
T \le C(\mathbf{p'}) + \frac{d}{1-\mu} A(\mathbf{p'}) + \left(1-\mu -\frac{\mu}{1-\mu}\right) T_2 \ .
\end{align*}
Then, applying the bounds for $C(\mathbf{p'})$ and $A(\mathbf{p'})$ in Lemma \ref{lem.allocation} from the resource allocation algorithm, and when $(1-\mu)^2 \le \mu$, i.e., $\mu \ge \frac{3-\sqrt{5}}{2}=1-\frac{1}{\phi}$, which makes the last term above at most zero, we can derive:
\begin{align*}
T \le \left(\frac{1}{\rho} + \frac{d}{(1-\mu)(1-\rho)}\right) T_{\opt} \triangleq f_d(\mu, \rho) \cdot T_{\opt} \ .
\end{align*}

Clearly, $f_{d}(\mu,\rho)$ is an increasing function of $\mu$ for all $d$. Thus, to minimize the function, we can set $\mu^{*}=1-\frac{1}{\phi}$. In this case, we require $P^{\min} \ge \frac{1}{(\mu^*)^2} \approx 6.854$ and we define $f_d(\rho) \triangleq f_d(\mu^*, \rho) = \frac{1}{\rho}+\frac{\phi d}{1-\rho}$. Now, by setting $f'_d(\rho) = -\frac{1}{\rho^2}+ \frac{\phi d}{(1-\rho)^2} = 0$ and by checking that $f''_d(\rho) > 0$ for all $\rho$, we get $\rho^* = \frac{1}{\sqrt{\phi d}+1}$ that minimizes $f_d(\rho)$. Thus, the approximation ratio is given by $f_d(\mu^*, \rho^*) = \phi d +2\sqrt{\phi d} +1$.
\end{proof}

We point out that, when there is only one type of resource (i.e., $d=1$), Theorem \ref{thm.main} gives an approximation ratio of 5.164, which improves upon the ratio of 5.236 by Lep\`{e}re et al. \cite{Lepere01_DAG}. Jansen and Zhang \cite{Jansen06_DAG} showed that the algorithm actually achieves an even better ratio of 4.73 by proving a tighter critical-path bound than the one shown in Lemma \ref{lem.cp_bound}. Unfortunately, their analysis cannot be generalized to the case with more than one type of resources.

While Theorem 1 proves the approximation ratio of the multi-resource scheduling algorithm for any $d$, the following theorem shows an improved result for large $d$.
\begin{theorem}\label{thm.main2}
For $d\ge 22$ and if $P^{\min} \ge d^{2/3}$, the performance of the multi-resource scheduling algorithm satisfies:
\begin{align*}
\frac{T}{T_{\opt}} \le d+3\sqrt[3]{d^2}+O(\sqrt[3]{d}) \ .
\end{align*}
The result is achieved at $\mu^* \approx \frac{1}{\sqrt[3]{d}}$ and $\rho^* = \frac{\sqrt{1-2\mu^*}}{\sqrt{1-2\mu^*} + \sqrt{d\mu^*}}$.
\end{theorem}

\begin{proof}
Following the proof of Theorem \ref{thm.main} but by substituting $T_2$ and $T_3$ into Equation (\ref{eq.T}), and if $P^{\min} \ge \frac{1}{\mu^2}$, we get:
\begin{align*}
T \le \frac{1-2\mu}{\mu(1-\mu)}C(\mathbf{p'}) + \frac{d}{1-\mu} A(\mathbf{p'}) + \left(1 -\frac{1-2\mu}{\mu(1-\mu)}\right) T_1 \ .
\end{align*}
Applying the bounds for $C(\mathbf{p'})$ and $A(\mathbf{p'})$ in Lemma \ref{lem.allocation}, and when $1 -\frac{1-2\mu}{\mu(1-\mu)} \le 0$, i.e., $\mu \le \frac{3-\sqrt{5}}{2}=1-\frac{1}{\phi}$, which makes the last term above at most zero, we can derive:
\begin{align*}
T \le \left(\frac{1-2\mu}{\mu(1-\mu)\rho} + \frac{d}{(1-\mu)(1-\rho)}\right) T_{\opt} \triangleq g_d(\mu, \rho) \cdot T_{\opt} \ .
\end{align*}

Let $X_{\mu}=\frac{1-2\mu}{\mu(1-\mu)}=\frac{1}{\mu}-\frac{1}{1-\mu}$ and $Y_{\mu}=\frac{1}{1-\mu}$. We can then write: $g_{d}(\mu,\rho) = \frac{X_{\mu}}{\rho}+\frac{dY_{\mu}}{1-\rho}$. By deriving $g_{d}(\mu,\rho)$ with respect to $\rho$ and setting the derivative to zero, we can get the best choice for $\rho$ to be $\rho^{*}(\mu)=\frac{\sqrt{X_{\mu}}}{\sqrt{X_{\mu}}+\sqrt{dY_{\mu}}}$. As $X_{\mu}, Y_{\mu} > 0$, clearly $\rho^{*}(\mu) \in (0, 1)$, thus is a valid choice. By substituting $\rho^{*}(\mu)$ back into $g_{d}(\mu,\rho)$ and simplifying, we can get:
\begin{align*}
g_{d}(\mu,\rho^{*}(\mu))=\big(\sqrt{X_{\mu}}+\sqrt{dY_{\mu}}\big)^{2} \triangleq g_{d}(\mu)^{2} \ .
\end{align*}

We will now minimize $g_{d}(\mu)=\sqrt{\frac{1}{\mu}-\frac{1}{1-\mu}}+\sqrt{\frac{d}{1-\mu}}$. By deriving $g_{d}(\mu)$ with respect to $\mu$ and factoring, we can get:
\begin{align*}
g'_{d}(\mu)=-\frac{(2d+4)\mu^{4}-(d+8)\mu^{3}+8\mu^{2}-4\mu+1}{2\mu(1-\mu)\sqrt{\mu(1-\mu)(1-2\mu)}\big(\mu\sqrt{d\mu(1-2\mu)}+(2\mu^{2}-2\mu+1)\big)} \ .
\end{align*}
As $2\mu^2-2\mu+1=\mu^2+(1-\mu)^2>0$ for any $\mu\in (0, 0.5)$, the denominator of $g'_{d}(\mu)$ is always positive. Thus, the sign of $g'_{d}(\mu)$ is the opposite of the sign of its numerator, which we define as:
\begin{align*}
h_{d}(\mu) \triangleq (2d+4)\mu^{4}-(d+8)\mu^{3}+8\mu^{2}-4\mu+1 \ .
\end{align*}

In the following, we will show that, if $d \leq 21$, $h_{d}(\mu)$ is always positive for any $\mu \in (0, \frac{3-\sqrt{5}}{2}]$, and thus the optimal choice is $\mu^{*}=\frac{3-\sqrt{5}}{2}$, which gives the same result as in Theorem \ref{thm.main}. Otherwise, if $d \ge 22$, there is a unique optimal choice $\mu^*\in (0, \frac{3-\sqrt{5}}{2})$, which satisfies $h_{d}(\mu^*)=0$. For convenience, we define $\mu^A = \frac{3-\sqrt{5}}{2}$ and $\mu^B = \frac{3}{8} < \mu^A$.

First, we can compute, for any $\mu \in (0, \mu^B]$, that:
\begin{align*}
h'_{d}(\mu) &= 4(2d+4)\mu^{3}-3(d+8)\mu^{2}+16\mu-4 \\
&= d\mu^{2}(8\mu-3)+4(2\mu-1)\big(\mu^{2}+(1-\mu)^{2}\big) < 0 \ .
\end{align*}

We can also compute, for any $\mu \in [\mu^B, \mu^A]$, that:
\begin{align*}
h''_{d}(\mu) &= 12(2d+4)\mu^{2}-6(d+8)\mu+16 \\
&\ge 12(2d+4)\cdot \Big(\frac{3}{8}\Big)^{2}-6(d+8)\cdot \Big(\frac{3-\sqrt{5}}{2}\Big)+16 \\
&\approx 1.083d + 4.416 > 0
\end{align*}

Thus, we can conclude the following:
\begin{itemize}
\item In $(0, \mu^B]$, $h_{d}(\mu)$ is a strictly decreasing function of $\mu$;
\item In $[\mu^B, \mu^A]$, $h_{d}(\mu)$ is a strictly convex function of $\mu$, and $h'_{d}(\mu)$ is a strictly increasing function of $\mu$.
\end{itemize}

We now distinguish two cases depending on the value of $d$.

Case (1): $d \le 21$. Since $h'_{d}(\mu)$ is an increasing function of $\mu$ in $[\mu^B, \mu^A]$, the largest value of $h_{d}'(\mu)$ is achieved at $\mu^A$. Also, $h_{d}'(\mu)$ is clearly an increasing function of $d$ for any $\mu> \mu^B$. Thus, for any $\mu\in (\mu^B, \mu^A]$, we have:
\begin{align*}
h'_{d}(\mu) \le h'_{d}(\mu^A) \le h'_{21}(\mu^A) \approx -0.328 < 0 \ .
\end{align*}

Thus, $h_d(\mu)$ is a strictly decreasing function of $\mu$ in $(0, \mu^A]$, and for any $\mu \in (0, \mu^A]$, if $d\le 21$, we have:
$$h_{d}(\mu) \ge h_{d}(\mu^{A}) \approx -0.013d+0.2786 \ge 0.0035 > 0 \ . $$

Since $g'_{d}(\mu)$ and $h_{d}(\mu)$ have opposite signs, this means $g'_{d}(\mu) < 0$, which implies that $g_d(\mu)$ is a decreasing function of $\mu$ in $(0, \mu^A]$. Therefore, the optimal $\mu$ to minimize $g_d(\mu)$ is $\mu^* = \mu^A = \frac{3-\sqrt{5}}{2}$. It can be verified that this choice yields the same approximation result as in Theorem \ref{thm.main}.

Case (2): $d \ge 22$. For any fixed $\mu$ in $(0,\mu^A]$, we can easily show that $h_{d}(\mu)$ is a decreasing function of $d$ (by deriving $h_{d}(\mu)$ with respect to $d$). Thus, we have $h_{d}(\mu^B) \leq h_{22}(\mu^B) \approx -0.008 <0$. Further, we have $h_{d}(0)=1>0$. Since $h_{d}(\mu)$ is a strictly decreasing function of $\mu$ in $(0,\mu^B]$, we know that $h_{d}(\mu)=0$ admits a unique solution $\mu^*$ in this interval.
%on this interval, is positive on $[0,\mu^{*}]$ and negative on $[\mu^{*},\frac{3}{8}]$.
Moreover, since $h_{d}(\mu)$ is a convex function in $[\mu^B, \mu^A]$, we have, for any $\mu \in [\mu^B, \mu^A]$, that:
\begin{align*}
h_{d}(\mu) \leq h_{22}(\mu) &\leq \max\big(h_{22}(\mu^B), h_{22}(\mu^{A})\big) \\
&\approx \max(-0.008,-0.01) < 0 \ .
\end{align*}

This shows that $h_{d}(\mu)>0$ in $(0,\mu^{*})$ and $h_{d}(\mu)<0$ in $(\mu^{*},\mu^{A}]$. Since $h_{d}(\mu)$ and $g'_{d}(\mu)$ have opposite signs, we get that $g_{d}(\mu)$ is a strictly decreasing function of $\mu$ in $(0,\mu^{*})$ and a strictly increasing function in $(\mu^{*},\mu^{A}]$. Thus, the optimal $\mu$ to minimize $g_d(\mu)$ is given by $\mu^*$.

As $\mu^{*}$ is the solution to a fourth-degree equation (i.e., $h_d(u) = 0$), its closed form, although exists, is too complicated to express. However, observing that when $d$ increases and if $\mu$ is small enough, the dominating negative term of $h_{d}(\mu)$ is $d\mu^{3}$ and the dominating positive term is $1$. We can then get an estimate of $\mu^{*}\approx \frac{1}{\sqrt[3]{d}}$, which gives an estimated approximation ratio: $g_d(\mu^*)^2\approx \frac{d\sqrt[3]{d}+2d\sqrt{1-\frac{2}{\sqrt[3]{d}}}+\sqrt[3]{d^{2}}-2\sqrt[3]{d}}{\sqrt[3]{d}-1} = d+3\sqrt[3]{d^2}+O(\sqrt[3]{d})$.
\end{proof}

Figure \ref{fig.ratio} plots the estimated ratio of Theorem \ref{thm.main2} in comparison with the actual ratio that results from the true value of $\mu^*$ (obtained numerically) for $22\le d \le 50$. We can see that the estimation is indeed very close to the actual value, and the result clearly improves upon the ratio of Theorem \ref{thm.main}.

\begin{figure}[h]
\center
\includegraphics[width=0.5\textwidth]{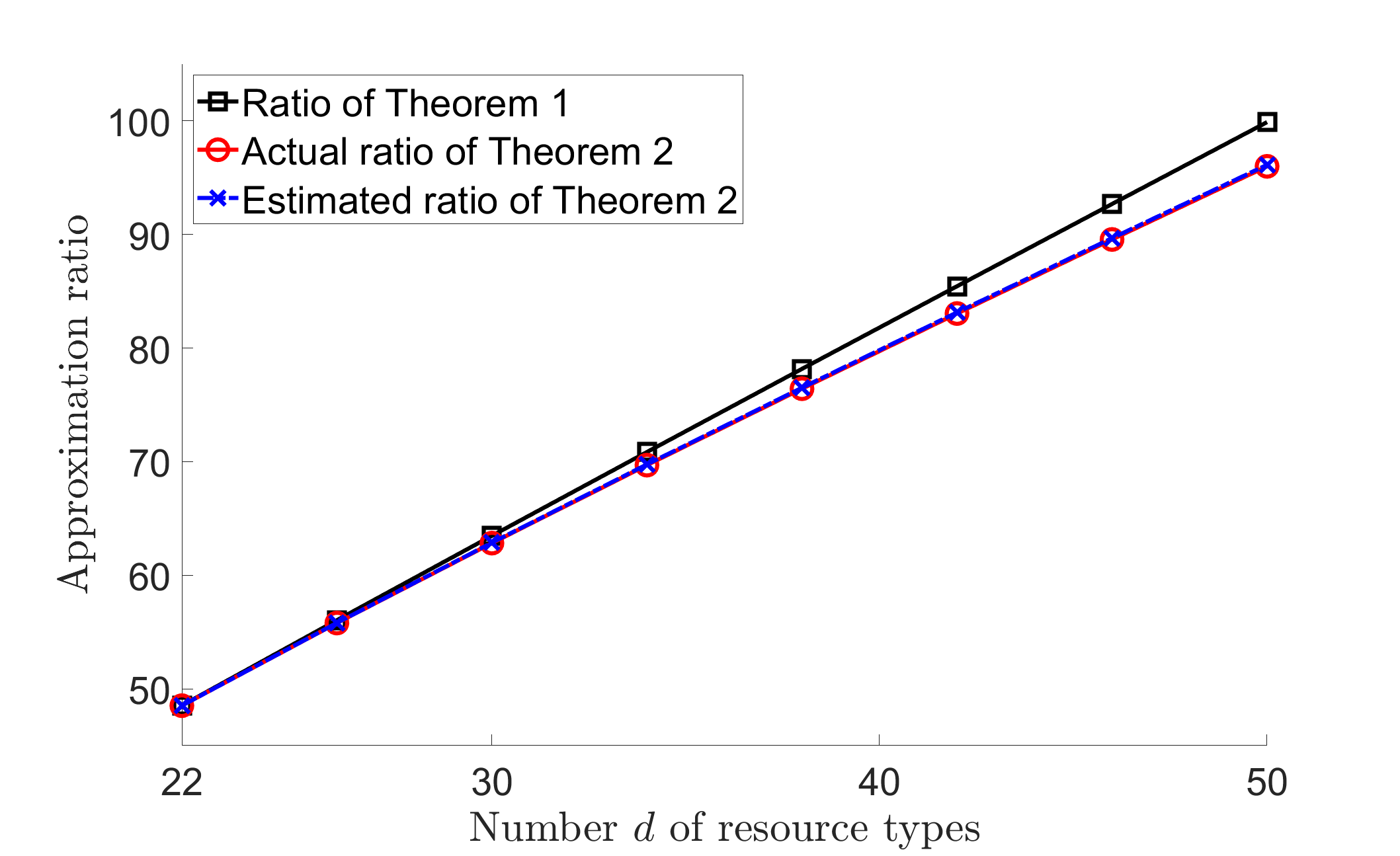}
	\caption{Comparison of the estimated ratio and the actual ratio of Theorem \ref{thm.main2} along with the ratio of Theorem \ref{thm.main} for $22\le d\le 50$.}\label{fig.ratio}
\end{figure}

Although Theorem \ref{thm.main2} holds for a large number of resource types (i.e., $d\ge 22$) and is unlikely to be practical in today's resource management systems, the result does have significant theoretical importance. In particular, it gives the first approximation for general list-based algorithm that is asymptotically tight up to the dominating factor $d$ in the context of multi-resource moldable job scheduling (see Theorem \ref{thm.lowerbound}).

\section{Improved Approximation Results for Some Special Graphs}\label{sec.improved}

In the preceding section, we have derived the approximation ratios of the multi-resource scheduling algorithm for general graphs.
In this section, we will show improved approximation results for some special graphs, namely, series-parallel graphs or trees, and independent jobs without any precedence constraints.

\subsection{Results for SP Graphs or Trees}

We first consider jobs whose precedence constraints form a series-parallel graph or a tree. A directed acyclic graph (DAG) is a \emph{series-parallel (SP) graph} \cite{Bodlaender96_SPdags} if it has only two nodes (i.e., a source and a sink) connected by an edge, or can be constructed (recursively) by a series composition or a parallel composition of two SP graphs.\footnote{Given two SP graphs $G_1$ and $G_2$, the \emph{parallel composition} is the union of the two graphs while merging their sources to create the new source and merging their sinks to create the new sink, and the \emph{series composition} merges the sink of $G_1$ with the source of $G_2$ and uses the source of $G_1$ as the new source and the sink of $G_2$ as the new sink.
}
Trees are simply special cases of general SP graphs.

In this case, we rely on an FPTAS (Fully Polynomial-Time Approximation Scheme) proposed in \cite{Lepere01_DAG} to find a near-optimal resource allocation.
The algorithm was proposed in the context of a single resource type, but can be readily adapted to work for multiple types of resources (by first discarding the subset of dominated resource allocations as shown in Step 1 of Algorithm \ref{alg.allocation}). In essence, the FPTAS first decomposes an SP graph into atomic parts, then uses dynamic programming to decide if an allocation $\mathbf{p'}$ that satisfies $L(\mathbf{p'}) \le X$ can be found for a positive
integer $X$, and finally performs a binary search on $X$. The following lemma shows the result. More details about the algorithm can be found in \cite{Lepere01_DAG}.

\begin{lemma}\label{lem.SP}
For a set of jobs whose precedence constraints form a series-parallel graph or a tree, and for any $\epsilon \ge 0$, an FPTAS (i.e., polynomial in $1/\epsilon$) exists, which can compute a resource allocation $\mathbf{p'} = (p'_1, p'_2,\dots,p'_n)$ that satisfies:
\begin{align*}
L(\mathbf{p'}) = \max(A(\mathbf{p'}), C(\mathbf{p'})) \le (1+\epsilon) \!\cdot\! L_{\min} \le (1+\epsilon) \!\cdot\! T_{\opt} \ .
\end{align*}
\end{lemma}

We can now use the above FPTAS to replace Step 2 in resource allocation (Algorithm \ref{alg.allocation}) and combine it with list scheduling (Algorithm \ref{alg.list}). The following theorem shows the approximation ratio for any number $d$ of resource types.

\begin{theorem}\label{thm.sp_main}
For any $d\ge 1$ and if $P^{\min} \ge 7$, the performance of the multi-resource scheduling algorithm for SP graphs or trees satisfies the following:
\begin{align*}
\frac{T}{T_{\opt}} \le (1+\epsilon)\cdot \left(\phi d+1 \right) \le (1+\epsilon)\cdot \left(1.619 d + 1 \right) \ ,
\end{align*}
where $\phi = \frac{1+\sqrt{5}}{2}$ is the golden ratio. The result is achieved at $\mu^* = 1-\frac{1}{\phi}\approx 0.382$.
\end{theorem}

\begin{proof}
Following the proof of Theorem \ref{thm.main} by substituting $T_1$ from Lemma (\ref{lem.cp_bound}) and $T_3$ from Lemma (\ref{lem.area_bound}) into $T = T_1 + T_2 + T_3$, and if $P^{\min} \ge \frac{1}{\mu^2}$, we get:
\begin{align*}
T \le C(\mathbf{p'}) + \frac{d}{1-\mu} A(\mathbf{p'}) + \left(1-\mu -\frac{\mu}{1-\mu}\right) T_2 \ .
\end{align*}

Then, by applying the bounds in Lemma \ref{lem.SP}, and when $(1-\mu)^2 \le \mu$, i.e., $\mu \ge \frac{3-\sqrt{5}}{2}=1-\frac{1}{\phi}$, we can derive:
\begin{align*}
T \le (1+\epsilon)\cdot\left(1 + \frac{d}{(1-\mu)}\right) T_{\opt} \triangleq f_d(\mu) \cdot T_{\opt} \ .
\end{align*}

Clearly, $f_d(\mu)$ is an increasing function of $\mu$ for all $d$. Thus, the minimum value is obtained by setting $\mu^{*}=1-\frac{1}{\phi}$. In this case, the approximation ratio is given by $f_d(\mu^*) = (1+\epsilon)\cdot \left(\phi d+1 \right)$, with the condition $P^{\min} \ge \frac{1}{(\mu^*)^2} \approx 6.854$.
\end{proof}

The approximation ratio can be improved with $d\ge 4$ resource types, as shown in the following theorem.

\begin{theorem}\label{thm.sp_main2}
For any $d\ge 4$ and if $P^{\min} \ge d + 2\sqrt{d-1}$, the performance of the multi-resource scheduling algorithm for SP graphs or trees satisfies the following:
\begin{align*}
\frac{T}{T_{\opt}} \le (1+\epsilon)\cdot \left(d+2\sqrt{d-1}\right) \ .
\end{align*}
The result is achieved at $\mu^* = \frac{1}{\sqrt{d-1}+1}$.
\end{theorem}

\begin{proof}
Following the proof of Theorem \ref{thm.main} but by substituting $T_2$ and $T_3$ into $T = T_1 + T_2 + T_3$, and if $P^{\min} \ge \frac{1}{\mu^2}$, we get:
\begin{align*}
T \le \frac{1-2\mu}{\mu(1-\mu)}C(\mathbf{p'}) + \frac{d}{1-\mu} A(\mathbf{p'}) + \left(1 -\frac{1-2\mu}{\mu(1-\mu)}\right) T_1 \ .
\end{align*}

Applying the bounds in Lemma \ref{lem.SP}, and when $1 -\frac{1-2\mu}{\mu(1-\mu)} \le 0$, i.e., $\mu \le \frac{3-\sqrt{5}}{2}$, we can derive:
\begin{align*}
T &\le (1+\epsilon)\cdot\left(\frac{1-2\mu}{\mu(1-\mu)} + \frac{d}{1-\mu}\right) T_{\opt} \\
&= (1+\epsilon)\cdot\left(\frac{1}{\mu} + \frac{d-1}{1-\mu}\right) \triangleq g_d(\mu) \cdot T_{\opt} \ .
\end{align*}

By setting $g'_{d}(\mu) = -\frac{1}{\mu^2} + \frac{d-1}{(1-\mu)^2} = 0$ and by checking that $g''_{d}(\mu) > 0$, we get $\mu^* = \frac{1}{\sqrt{d-1}+1}$, which is at most $\frac{3-\sqrt{5}}{2}$ for $d\ge 4$. Thus, with the condition $P^{\min} \ge \frac{1}{(\mu^*)^2} = d + 2\sqrt{d-1}$ and $d\ge 4$, we get the approximation ratio:
\begin{align*}
g_{d}(\mu^*) &= (1+\epsilon)\cdot \left(\sqrt{d-1}+1 + \frac{d-1}{1-\frac{1}{\sqrt{d-1}+1}}\right) \\
&= (1+\epsilon) \cdot \left(d+2\sqrt{d-1}\right) \ .   \qedhere
\end{align*}
\end{proof}

\subsection{Results for Independent Jobs}

We finally consider independent jobs without any precedence constraints. For this case, Sun et al. \cite{Sun18_multiresource} presented a $2d$-approximation algorithm for any $d\ge 1$, while we show improved results for $d \ge 3$.
Here, we rely on an optimal multi-resource allocation algorithm proposed in \cite{Sun18_multiresource} as Step 2 of our Algorithm \ref{alg.allocation}. The algorithm computes the resource allocation in polynomial time as shown in the lemma below. More details of the algorithm can be found in \cite{Sun18_multiresource}.

\begin{lemma}\label{lem.independent}
For a set of independent jobs, a resource allocation $\mathbf{p'} = (p'_1, p'_2,\dots,p'_n)$ can be found in polynomial time, such that:
\begin{align*}
L(\mathbf{p'}) = \max(A(\mathbf{p'}), C(\mathbf{p'})) = L_{\min} \le T_{\opt} \ ,
\end{align*}
where $C(\allresources') = \max_{j=1\dots n} t_j(p'_j)$ denotes the maximum execution time of any job under allocation $\allresources'$, which becomes the critical path when there is no precedence constraint.
\end{lemma}

For independent jobs, while the area bound (Lemma \ref{lem.area_bound}) remains unchanged, we show a modified critical-path bound.
\begin{lemma}[Modified Critical-Path Bound]\label{lem.modifiedcp_bound}
For any choice of $\mu \in (0, 0.5)$, we have:
\begin{itemize}
\item If $\mathcal{I}_1 = \emptyset$, $\mu T_2 \le C(\mathbf{p'})$;
\item If $\mathcal{I}_1 \neq \emptyset$, $T_1 + T_2 \le C(\mathbf{p'})$.
\end{itemize}
\end{lemma}

\begin{proof}
Recall that there are three categories of intervals $\mathcal{I}_1$, $\mathcal{I}_2$ and $\mathcal{I}_3$.
Based on the proof of Lemma \ref{lem.cp_bound}, during any interval $I \in \mathcal{I}_1 \cup \mathcal{I}_2$, there is no ready job in the queue. Since all jobs are independent, it means that all jobs have been scheduled. This implies that all intervals in $\mathcal{I}_2$ happen before all intervals in $\mathcal{I}_1$, since there is no new job arrival and jobs only complete. Further, all intervals in $\mathcal{I}_3$ happen before all intervals in $\mathcal{I}_2$ using the same argument. Now, consider a job $j$ that completes the last in the schedule. We know that $j$ must have started during $\mathcal{I}_3$ or at the beginning of $\mathcal{I}_2$. We consider two cases.

Case (1): $\mathcal{I}_1 = \emptyset$. In this case, job $j$ is executed during all intervals in $\mathcal{I}_2$ and it could be adjusted. Thus, according to Lemma~\ref{lem.adjust} (Inequality (\ref{eq.tj})), we have $\mu T_2 \le \mu \cdot t_{j}(p_{j}) \le t_{j}(p'_{j}) \le \max_{j=1\dots n} t_j(p'_j) = C(\allresources')$.

Case (2): $\mathcal{I}_1 \neq \emptyset$. In this case, job $j$ is executed during all intervals in $\mathcal{I}_2$ as well as all intervals in $\mathcal{I}_1$. Thus, job $j$ must be unadjusted (since it is executed during $\mathcal{I}_1$). Thus, we have $T_1 + T_2 \le t_{j}(p_{j}) = t_{j}(p'_{j}) \le \max_{j=1\dots n} t_j(p'_j) = C(\allresources')$.
\end{proof}

\begin{theorem}\label{thm.independent}
The performance of multi-resource scheduling for independent jobs satisfies $T/T_{\opt} \le r$, where:
\begin{align*}
r = \begin{cases}
2d,  & \text{~if~} d = 1, 2, \text{~and~} P^{\min} \ge 1 \\
1.619d+1, & \text{~if~} d = 3, \text{~and~} P^{\min} \ge 7 \\
d+2\sqrt{d-1}, & \text{~if~} d \ge 4, \text{~and~} P^{\min} \ge d+2\sqrt{d-1}
\end{cases}
\end{align*}
\end{theorem}

\begin{proof}
When $d = 1, 2$, we can just apply the multi-resource scheduling algorithm in \cite{Sun18_multiresource} to get $2d$-approximation. Otherwise, we consider both cases as stated in Lemma \ref{lem.modifiedcp_bound}.

Case (1): $\mathcal{I}_1 = \emptyset$. In this case, the makespan is given by $T = T_2 + T_3$. Substituting $\mu T_2 \le C(\mathbf{p'})$ from Lemma \ref{lem.modifiedcp_bound} and $\mu T_2 + (1-\mu) T_3 \le d\cdot A(\mathbf{p'})$ from Lemma \ref{lem.area_bound} into $T$, we get:
\begin{align*}
T &\le \frac{1-2\mu}{\mu(1-\mu)} C(\mathbf{p'}) + \frac{d}{1-\mu} A(\mathbf{p'}) \\
&\le \Big( \frac{1-2\mu}{\mu(1-\mu)} + \frac{d}{1-\mu} \Big) \cdot T_{\opt} \quad\quad (\text{by Lemma \ref{lem.independent}}) \\
&\triangleq g_d(\mu) \cdot T_{\opt} \ .
\end{align*}

Case (2): $\mathcal{I}_1 \neq \emptyset$. In this case, the makespan is given by $T = T_1 + T_2 + T_3$. Substituting $T_1 + T_2 \le C(\mathbf{p'})$ from Lemma \ref{lem.modifiedcp_bound} and $\mu T_2 + (1-\mu) T_3 \le d\cdot A(\mathbf{p'})$ from Lemma \ref{lem.area_bound} into $T$, we get:
\begin{align*}
T &\le C(\mathbf{p'}) + \frac{d}{1-\mu} A(\mathbf{p'}) - \frac{\mu}{1-\mu} T_2 \\
&\le \Big( 1 + \frac{d}{1-\mu} \Big) \cdot T_{\opt} \quad\quad (\text{by Lemma \ref{lem.independent}}) \\
&\triangleq f_d(\mu) \cdot T_{\opt} \ .
\end{align*}

The overall approximation ratio is given by $\max(f_d(\mu), g_d(\mu))$, with the condition $P^{\min} \ge \frac{1}{\mu^2}$. Thus,
when $d=3$, by following the proof of Theorem \ref{thm.sp_main} and setting $\mu^* \approx 0.382$, the ratio is $f_d(\mu^*) \le 1.619d+1$.
When $d\ge 4$, we can follow the proof of Theorem \ref{thm.sp_main2} by setting $\mu^* = \frac{1}{\sqrt{d-1}+1}$. In this case, the ratio is $g_d(\mu^*) = d+2\sqrt{d-1}$.
\end{proof}

\section{Lower Bound for List Scheduling}\label{sec.listlowerbound}

Lastly, we prove a lower bound of $d$ on the approximation ratio of any deterministic algorithm that, for the second phase, uses list scheduling with only \emph{local} priority considerations (i.e., without taking into account the precedence graphs when assigning priorities to the jobs). This lower bound holds regardless of the resource allocation scheme for the first phase. The result shows that our multi-resource scheduling algorithms essentially achieve tight approximation ratios up to the dominating factor for large $d$ among the generic class of local list scheduling schemes.

\begin{figure}
\center
\includegraphics[height=9cm]{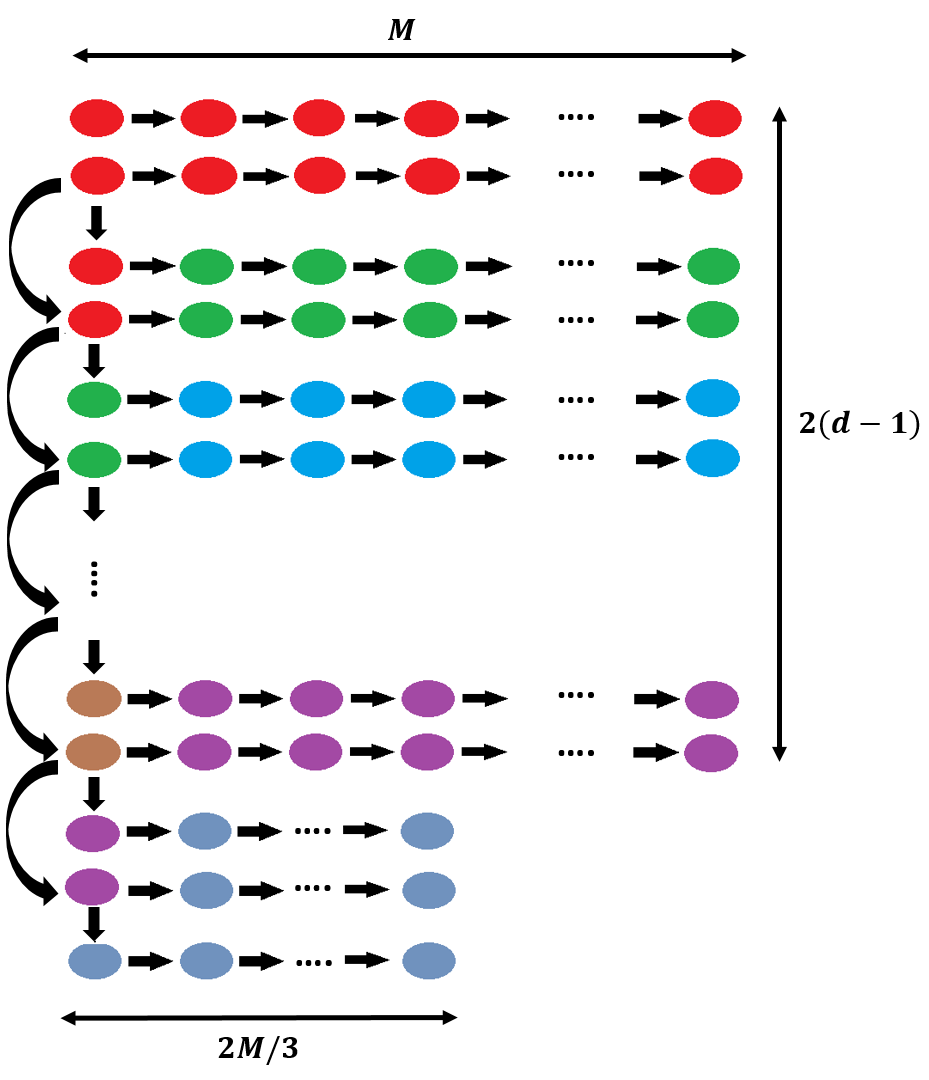}
\caption{Lower bound instance with an approximation ratio of $d$ for any deterministic list scheduling algorithm with local job priority considerations.}\label{fig.tight}
\end{figure}

\begin{theorem}\label{thm.lowerbound}
Any deterministic list scheduling algorithm with local job priority considerations is no better than $d$-approximation for the multi-resource scheduling problem.
\end{theorem}

\begin{proof}
The lower bound is constructed by using a set of jobs whose precedence constraints form a tree. Each job takes unit-time to complete, and only requires a unit resource allocation from a single resource type.
For each resource type $i$, there is a total amount $P^{(i)} = 2$ of available resource.
Figure \ref{fig.tight} illustrates our lower bound instance with $n = 2Md$ jobs, where $M$ is an integer multiple of 3. The nodes represent the jobs, the arrows represent the precedence constraints, and the color of a node represents the single resource type the corresponding job requires.

The optimal schedule can be obtained by prioritizing the job dependencies going downward, resulting in a makespan of $T_{\opt} = M + d - 1$.
Any deterministic list scheduling algorithm with only local priority considerations cannot distinguish jobs that require the same resource type. Hence, in the worst-case, it could only utilize one type of resource at any time, resulting in a makespan of $T = M(d-1) + \frac{4M}{3} = Md + \frac{M}{3}$.
%jobs are processed two lines by two lines, nothing will be parallelized and the total processing time would be $K(d-1)+\frac{4}{3}K$. Indeed, the time required to start the last three lines would be $K(d-1)$ in that case, then the algorithm would need $\frac{2}{3}K$ to process lines $2d-1$ and $2d$ and finally an additional $\frac{2}{3}K$ for the last line.  As all the jobs are indistinguishable, any local deterministic algorithm could make that choice.
Choosing $M > 3(d^2-d)$, the worst-case approximation ratio is:
\begin{align*}
\frac{T}{T_{\opt}} = \frac{Md+\frac{M}{3}}{M+d-1} = \frac{d+\frac{1}{3}}{1+\frac{d-1}{M}} > \frac{d+\frac{1}{3}}{1+\frac{1}{3d}} = d \ .
\end{align*}
This completes the proof of the theorem.
\end{proof}

\section{Conclusion}\label{sec.conclusion}

In this paper, we have studied the problem of scheduling parallel jobs with precedence constraints under multiple types of schedulable resources. We focused on moldable jobs, which allow the scheduler to flexibly select a variable set of resources before the execution of the jobs, and the goal is to minimize the overall completion time, or the makespan. We have proposed a multi-resource scheduling algorithm that adopts the two-phase approach by combining an approximate
resource allocation and an extended list scheduling scheme. We have proven approximation ratios of the algorithm for the general precedence graph, as well as for some special graphs including SP-DAGs or trees and independent jobs. The results are summarized in Table \ref{tab.summary}. We have also proven a lower bound on the approximation ratio of any local list scheduling scheme, which shows that our algorithm achieves the optimal asymptotic performance up to the dominating factor.

We point out that the lower bound proven in Theorem \ref{thm.lowerbound} does not rule out the possibility of a \emph{global} list scheduling algorithm that considers the structure of the precedence graph when determining the priorities for the jobs (e.g., giving priority to the jobs on the critical path). It remains an open question to find such an algorithm by showing a better approximation ratio than $d$, or to prove a matching lower bound for \emph{any} list-based scheduling scheme.

\begin{table}
\center
\caption{Summary of approximation results.}\label{tab.summary}
\begin{tabular}{|c|l|}
\hline
\tabincell{c}{Precedence} & ~~~~~~~~~Approximation Ratio \\
\hline
\hline
\tabincell{c}{General \\ Graphs} & \tabincell{l}{$\bullet$ $1.619d+2.545\sqrt{d}+1$ for $d\ge 1$ \\ $\bullet$ $d+3\sqrt[3]{d^2}+O(\sqrt[3]{d})$ for $d\ge 22$} \\
\hline
\tabincell{c}{SP Graphs \\ or Trees} & \tabincell{l}{$\bullet$ $(1+\epsilon) \left(1.619 d + 1 \right)$ for $d\ge 1$ \\  $\bullet$ $(1+\epsilon) \left(d+2\sqrt{d-1}\right)$ for $d\ge 4$} \\
\hline
\tabincell{c}{Independent \\ Jobs} & \tabincell{l}{$\bullet$ $2d$ for $d\ge 1$ \cite{Sun18_multiresource} \\ $\bullet$ $1.619d+1$ for $d = 3$ \\ $\bullet$ $d+2\sqrt{d-1}$ for $d\ge 4$} \\
\hline
\end{tabular}
\end{table}

\bibliographystyle{abbrv}

\end{document}